\documentclass[aps,pra,10pt,tightenlines,longbibliography,eqsecnum,notitlepage,amsmath,twocolumn,floatfix,superscriptaddress]{revtex4-1}
\usepackage[dvipsnames]{xcolor}
\usepackage[colorlinks=true,bookmarks=false,linkcolor=NavyBlue,urlcolor=NavyBlue,citecolor=NavyBlue,breaklinks]{hyperref}
\usepackage{amsmath}
\usepackage{amsfonts}
\usepackage{tabularx}
\usepackage{graphicx}
\usepackage{times}
\usepackage{mathtools}
\usepackage{braket}
\usepackage{amsthm}

\newtheorem{theorem}{Theorem}

\newtheorem{proposition}{Proposition}

\theoremstyle{remark}

\newcommand{\parti}[2]{\frac{\partial #1}{\partial #2}}

\newcommand{\intall}{\int_{-\infty}^{\infty}}
\newcommand{\avg}[1]{\langle#1\rangle}
\newcommand{\Avg}[1]{\left\langle#1\right\rangle}

\newcommand{\abs}[1]{\left|#1\right|}
\newcommand{\bk}[1]{\left(#1\right)}
\newcommand{\Bk}[1]{\left[#1\right]}
\newcommand{\BK}[1]{\left\{#1\right\}}
\newcommand{\trace}{\operatorname{tr}}

\newcommand{\expect}{\mathbb E}

\newcommand{\norm}[1]{\lVert#1\rVert}
\newcommand{\prior}{\rho}
\newcommand{\dop}{\varrho}
\newcommand{\range}{\operatorname{range}}

\begin{document}

\title{Physics-inspired forms of the Bayesian Cram\'er-Rao bound}

\author{Mankei Tsang}
\email{mankei@nus.edu.sg}
\homepage{https://blog.nus.edu.sg/mankei/}
\affiliation{Department of Electrical and Computer Engineering,
  National University of Singapore, 4 Engineering Drive 3, Singapore
  117583}

\affiliation{Department of Physics, National University of Singapore,
  2 Science Drive 3, Singapore 117551}

\date{\today}


\begin{abstract}
  Using differential geometry, I derive a form of the Bayesian
  Cram\'er-Rao bound that remains invariant under reparametrization.
  With the invariant formulation at hand, I find the optimal and
  naturally invariant bound among the Gill-Levit family of bounds. By
  assuming that the prior probability density is the square of a
  wavefunction, I also express the bounds in terms of functionals that
  are quadratic with respect to the wavefunction and its gradient. The
  problem of finding an unfavorable prior to tighten the bound for
  minimax estimation is shown, in a special case, to be equivalent to
  finding the ground state of a Schr\"odinger equation, with the
  Fisher information playing the role of the potential.  To illustrate
  the theory, two quantum estimation problems, namely, optomechanical
  waveform estimation and subdiffraction incoherent optical imaging,
  are discussed.
\end{abstract}

\maketitle
\section{Introduction}
Differential geometry has been useful in the study of statistical
divergence measures, Cram\'er-Rao bounds, and asymptotic statistics
\cite{amari,*amari16,bickel93,tsang20}, but its usefulness for
Bayesian and minimax statistics is less clear. The Bayesian
Cram\'er-Rao bounds \cite{bell,gill95}, pioneered by Sch\"utzenberger
\cite{schutzenberger57} and Van Trees \cite{vantrees}, may serve as a
bridge.

To set the stage, consider a $p$-dimensional parameter
$\theta = (\theta^1,\dots,\theta^p) \in \Theta \subseteq \mathbb R^p$,
a scalar parameter of interest $\beta(\theta) \in \mathbb R$ that is a
function of $\theta$, and an estimator $\check\beta(X)$, where $X$ is
a set of $n$ independent and identically distributed (i.i.d.)
observation random variables with a family of probability densities
$\{f^{(n)}(x|\theta) = \prod_{j=1}^n f(x^j|\theta):\theta \in\Theta\}$
and a reference measure $\mu$ that gives
$d\mu^{(n)}(x) = \prod_{j=1}^n d\mu(x_j)$.  Generalization of the
theory for a vectoral $\beta$ is straightforward but tedious and
deferred to Appendix~\ref{sec_vectoral}. Define the mean-square risk
as
\begin{align}
\mathsf R(\theta) &\equiv 
\int \Bk{\check\beta(x)-\beta(\theta)}^2 f^{(n)}(x|\theta) d\mu^{(n)}(x).
\end{align}
The Cram\'er-Rao bound for
any unbiased estimator is given by
\begin{align}
\mathsf R(\theta) &\ge \frac{\mathsf C(\theta)}{n},
\label{CRB_unbiased}
\\
\mathsf C(\theta) &\equiv u_a(\theta)\Bk{F(\theta)^{-1}}^{ab} u_b(\theta),
\label{CRB}
\end{align}
where 
\begin{align}
u_a &\equiv \partial_a\beta,
&
\partial_a &\equiv \parti{}{\theta^a},
\end{align}
Einstein summation is assumed, $F$ is the Fisher information
matrix defined as
\begin{align}
F_{ab} &\equiv \int \bk{\partial_a \ln f}
\bk{\partial_b \ln f} f d\mu,
\end{align}
$F^{-1}$ is its inverse such that $F_{ab}(F^{-1})^{bc} = \delta_a^c$,
and $\delta$ is the Kronecker delta. For simplicity, hereafter I call
Eqs.~(\ref{CRB_unbiased}) and (\ref{CRB}) the local bound, and
  the theory concerning $\mathsf C(\theta)$ the local theory, as
  $\mathsf C(\theta)$ depends only on the local properties of the
  statistical model in the neighborhood of $\theta$.

The restriction to unbiased estimators is one of the biggest
shortcomings of the local bound. A fruitful remedy is to consider
bounds on the Bayesian risk
\begin{align}
\Avg{\mathsf R} &= 
\expect\Bk{\bk{\check\beta-\beta}^2} = 
\int \mathsf R(\theta)\pi(\theta) d^p\theta,
\label{prior}
\end{align}
where $\expect$ denotes the expectation over both the observation and
the parameter as random variables and $\pi$ is a prior probability
density \cite{bell}.  In particular, Gill and Levit proposed a general
family of Bayesian Cram\'er-Rao bounds, valid for any biased or
unbiased estimator, given by \cite{gill95}
\begin{align}
\Avg{\mathsf R} &\ge \mathsf B
\equiv \frac{\Avg{\mathsf A}^2}{n\Avg{\mathsf F} + \Avg{\mathsf P}},
\label{GL}
\\
\mathsf A &\equiv v^au_a,
\label{A}
\\
\mathsf F &\equiv v^a F_{ab} v^b,
\\
\mathsf P &\equiv \Bk{\frac{1}{\pi}\partial_a\bk{\pi v^a}}^2,
\label{P}
\end{align}
where $v$, $\mathsf A$, $\mathsf F$, and $\mathsf P$ are all functions
of $\theta$, $\pi v$ is assumed to vanish on the boundary of $\Theta$,
and $\avg{\cdot}$ denotes the prior expectation, as in
Eq.~(\ref{prior}).

This work studies only the bound $\mathsf B$; the attainability of the
bound is outside the scope of this work.  Some recently proposed
Bayesian Cram\'er-Rao bounds \cite{bacharach19} may not fall under the
Gill-Levit family and are also outside the scope of this work. There
also exist many other types of Bayesian bounds that may be tighter,
such as the Ziv-Zakai bounds and the Weiss-Weinstein bounds
\cite{bell}, but the Cram\'er-Rao bounds are often much easier to
compute because they are based on the Fisher information, a
well-studied quantity.

  In Eqs.~(\ref{A})--(\ref{P}), $v:\mathbb R^p \to \mathbb R^p$
  is a free term, and by choosing it judiciously, many useful forms of
  $\mathsf B$ can be obtained \cite{gill95}. An arbitrarily chosen
  $v$, however, may lead to a $\mathsf B$ that varies if the
  parametrization of the underlying model with respect to $\theta$ is
  changed. To give a simple example, suppose that $p = 1$, $\theta$ is
  a scalar, and $\beta = \theta$. Consider the Gill-Levit bound for
  $v^1 = 1$.  If the parametrization of the underlying statistical
  model is changed, say, via the relation
  $\theta = \tilde\theta^{1/3}$ and $\tilde\theta = \theta^3$, then
  the Gill-Levit bound for $\beta = \theta = \tilde\theta^{1/3}$ and
  $v^1 = 1$ would usually be different when computed with respect to
  the new parameter $\tilde\theta$, even if the statistical problem
  remains the same.  This property is unpleasant, as there can be
  infinitely many parametrizations for the same model and it is not
  clear which parametrization leads to the tightest bound for a given
  problem. Note that the local bound given by Eq.~(\ref{CRB}) does not
  suffer from such a problem, as it is well known to be invariant upon
  reparametrization \cite{stein56}.  In Sec.~\ref{sec_invar}, I
  propose a condition on $v$ that makes $\mathsf B$ invariant. I also
  derive an invariant form of $\mathsf B$ using the language of
  differential geometry \cite{carroll}. With the invariant form,
  $\mathsf B$ is guaranteed to give the same value for a model,
  regardless of the parametrization.

A related question is how $v$ should be chosen. Although Gill and
Levit suggested a few options based on prior works or convenience, it
is unclear which is better, or if there exists an optimal choice.  In
Sec.~\ref{sec_optimal}, I show that there is indeed an optimal choice,
and it agrees with a couple of popular options in special cases. The
inspiration comes from the geometric picture of $v$ as a vector field,
which generalizes the role of a tangent vector in the local theory
\cite{bickel93,stein56}. By virtue of the invariant formalism, the
resultant bound is naturally invariant.

Bayesian bounds are also useful for
minimax statistics \cite{tsybakov} by providing lower bounds on the
worst-case risk via
\begin{align}
\sup_{\theta\in\Theta}\mathsf R(\theta) &\ge \Avg{\mathsf R}
\label{minimax}
\end{align}
for any prior. In this context, one should no longer choose the
  prior according to Bayesian principles. Instead, one should choose
  an unfavorable prior as a mathematical device to tighten a lower
  bound.  Given Eqs.~(\ref{GL})--(\ref{P}), it is unclear how the
prior should be chosen, as $\avg{\mathsf P}$ is highly nonlinear with
respect to $\pi$. To help with this problem, in Sec.~\ref{sec_wave} I
rewrite Eqs.~(\ref{GL})--(\ref{P}) in a form that looks more familiar,
at least to physicists. To be specific, I identify the prior density
with the square of a wavefunction, such that $\avg{\mathsf A}$,
$\avg{\mathsf F}$, and most importantly $\avg{\mathsf P}$ all become
quadratic functionals of the wavefunction and its gradient. In a
special case, $n\avg{\mathsf F} + \avg{\mathsf P}$ becomes the average
energy of a wave that obeys a Schr\"odinger equation. Finding the
tightest bound for minimax estimation then becomes equivalent to
finding the ground-state energy of the wave, and insights from quantum
mechanics turn out to be handy.

In terms of other prior works, Refs.~\cite{jupp10,kumar18} also study
Bayesian Cram\'er-Rao bounds in geometric terms, but do not discuss
the question of invariance or find the optimal Gill-Levit
bound. References~\cite{abushanab15,*koike20} derive the
asymptotically optimal form of the Gill-Levit bounds, but do not find
the exact optimal form. Example~4.2 in Ref.~\cite{bobrovsky} studies
the optimization of a Bayesian Cram\'er-Rao bound for a special
problem, but not in the generality considered here.  Regarding the
wave picture, the fact that $F$ is quadratic with respect to
$\partial_a(f^{1/2})$ is well known in statistics \cite{amari}, and
Frieden even claimed that it serves as a fundamental principle for
physics \cite{frieden}. He assumed that $f$ is the square of a
wavefunction and derived wave equations from this fact, but had to
introduce further creative assumptions. He also did not consider
Bayesian bounds.  To my knowledge, the wave picture of a Bayesian
Cram\'er-Rao bound is first proposed in Ref.~\cite{tsang18}, which
considers the special case $\beta = \theta$ with a scalar $\theta$ and
uses the wave picture as a trick to solve a parameter-estimation
problem in optical imaging. Here, as before \cite{tsang18}, I do not
claim that my results have any foundational implications for physics,
merely that the correspondence is interesting and useful for
statistics problems.

Section~\ref{sec_quantum} comes full circle and applies the
statistical theory to quantum estimation \cite{helstrom,hayashi},
where actual quantum systems are considered. I consider two important
problems in quantum optics, namely, optomechanical waveform estimation
\cite{braginsky,twc} and subdiffraction incoherent imaging
\cite{helstrom,tnl,*tsang19a}. The first problem is relevant to
gravitational-wave detectors, where quantum noise is now playing a
major role \cite{ligo16c,*miao17,*tse19,*acernese19,*yu20}; I show the
importance of including prior information in deriving a meaningful
quantum limit in terms of spectral quantities, following
Ref.~\cite{twc}. The second problem is, of course, a fundamental one
in optics and relevant to both fluorescence microscopy and
observational astronomy. Recent studies, based on quantum estimation
theory, have shown that judicious measurements can substantially
improve the imaging of subdiffraction objects \cite{tnl}, although
most prior works are based on the local bound, which is valid for
unbiased estimators only. By considering the minimax perspective, the
Bayesian bound, and the wave picture, I discuss the implication of a
zero information for the estimator convergence rate for the
multi-source localization problem studied in
Refs.~\cite{tnl,paur18,*paur19,bisketzi19}.

\section{\label{sec_invar}Invariance}
To model reparametrization, consider a bijective differentiable map
$\tilde\theta(\theta)$.  The transformation laws are
\begin{align}
\partial_a &=  J_a^b {\tilde \partial}_b,
&
{\tilde \partial}_a  &\equiv \parti{}{{\tilde\theta}^a},
\\
d^p\theta &= \frac{d^p\tilde\theta}{\norm{J}},
&
\pi &= \norm{J} \tilde\pi,
\\
u_a &= J_a^b \tilde u_b,
&
F_{ab} &= J_a^c {\tilde F}_{cd} J_b^d,
\label{covariant}
\end{align}
where 
\begin{align}
J_a^b &\equiv \partial_a\tilde\theta^b
\end{align}
is the Jacobian matrix, $|J|$ denotes its determinant, and $\norm{J}$
denotes the absolute value of the determinant.
Equations~(\ref{covariant}) imply that the components of $u$ are
covariant and $F$ is a $(0,2)$ tensor. On the other hand, $\beta$,
$\check\beta$, $f$, $\mu$, $\mathsf R$, and $\avg{\cdot}$ remain
invariant in the sense that 
  $\beta(\theta) = \tilde\beta(\tilde\theta(\theta))$,
  $f(x|\theta) = \tilde f(x|\tilde\theta(\theta))$,
  $\mathsf R(\theta) = \tilde{\mathsf R}(\tilde\theta(\theta))$,
  etc., as these quantities depend on the statistical problem and
should not depend on the parametrization of the underlying model.

It is well known that the local bound is invariant under
reparametrization \cite{stein56}, in the sense of
\begin{align}
u_a\bk{F^{-1}}^{ab} u_b = 
{\tilde u}_a \big({\tilde F}^{-1}\big)^{ab} {\tilde u}_b.
\end{align}
The Gill-Levit bounds can also be made invariant.

\begin{proposition}
  $\mathsf B$ is invariant under reparametrization if $v$ obeys the
  transformation law
\begin{align}
v^a  J_a^b &=  \tilde v^b.
\label{contra}
\end{align}
\label{prop_contra}
\end{proposition}
\begin{proof}
Given Eq.~(\ref{contra}), it is obvious that
\begin{align}
\mathsf A &= \tilde v^a \tilde u_a,
&
\mathsf F &= \tilde v^a \tilde F_{ab} \tilde v^b
\end{align}
remain invariant upon reparametrization. To deal with $\mathsf P$,
define the inverse Jacobian matrix  as
\begin{align}
\tilde J_a^b &\equiv \tilde\partial_a\theta^b,
\end{align}
which obey
\begin{align}
J_a^b\tilde J_b^c &= \tilde J_a^b J_b^c = \delta_a^c,
&
|\tilde J| &= \frac{1}{|J|}.
\end{align}
Consider
\begin{align}
\frac{1}{\pi}\partial_a\bk{\pi v^a} 
&= \frac{J_a^b}{|J|\tilde\pi}
\tilde\partial_b\bk{|J|\tilde\pi \tilde v^c\tilde J_c^a}
\\
&= \tilde v^c
\frac{J_a^b}{|J|}\tilde\partial_b\bk{|J|\tilde J_c^a}
+ \frac{1}{\tilde\pi}\tilde\partial_b \bk{\tilde\pi\tilde v^b}.
\end{align}
The first term can be shown to vanish as follows:
\begin{align}
\frac{J_a^b}{|J|}\tilde\partial_b\bk{|J|\tilde J_c^a}
&= \tilde\partial_c \ln |J|+ J_a^b\tilde\partial_b \tilde J_c^a
\\
&= -\tilde\partial_c \ln |\tilde J|
+ J_a^b\tilde\partial_b \tilde J_c^a
\\
&= - J_a^b \tilde\partial_c \tilde J_b^a
+ J_a^b\tilde\partial_b \tilde J_c^a
\label{formula}
\\
&= -J_a^b \bk{\tilde\partial_c \tilde\partial_b \theta^a
-\tilde\partial_b\tilde\partial_c\theta^a} = 0,
\end{align}
where Eq.~(\ref{formula}) uses Jacobi's formula to simplify
$\tilde\partial_c \ln|\tilde J|$. Hence
\begin{align}
\frac{1}{\pi}\partial_a\bk{\pi v^a} 
&=  \frac{1}{\tilde\pi}\tilde\partial_b \bk{\tilde\pi\tilde v^b},
\end{align}
and $\mathsf P$ is invariant. As the prior expectation $\avg{\cdot}$
is also invariant, $\mathsf B$ is invariant.
\end{proof}

In the language of differential geometry, Eq.~(\ref{contra}) means
that the components of $v$ are contravariant. In other words, $v$
defines a vector field in the parameter space $\Theta$, with
components $(v^1,\dots,v^p)$ with respect to a parametrization. If one
does not transform the components as per Eq.~(\ref{contra}) upon
reparametrization, $\mathsf B$ changes---the reason, from the
geometric perspective, is that it has become a bound for a different
vector field. For someone familiar with differential geometry,
Prop.~\ref{prop_contra} may seem trivial in hindsight, but this
triviality should be regarded as a virtue---it is evidence that
differential geometry is useful in simplifying the problem here.

A ``natural'' choice of the $v$ components according to Gill and
Levit is \cite{gill95}
\begin{align}
v^a &= (F^{-1})^{ab}u_b.
\label{natural}
\end{align}
This form is contravariant, in the sense that Eq.~(\ref{natural})
  for one parametrization and
  $\tilde v^a = (\tilde F^{-1})^{ab}\tilde u_b$ for another
  parametrization obey Eq.~(\ref{contra}) and must give the same bound
  for a given problem.  This choice also leads to the simplification
\begin{align}
\mathsf A &= \mathsf F = u_a(F^{-1})^{ab}u_b = \mathsf C,
\end{align}
which coincides with the local bound given by Eq.~(\ref{CRB}). The
resultant Bayesian bound is
\begin{align}
\mathsf B &= \frac{\avg{\mathsf C}^2}{n\avg{\mathsf C} + \avg{\mathsf P}}.
\label{borovkov}
\end{align}
For a scalar $\theta$, this becomes an inequality of Borovkov and
Sakhanenko \cite{borovkov80}; see also Ref.~\cite{borovkov}.  Most
importantly, Eq.~(\ref{borovkov}) agrees with some classic theorems in
the asymptotic local theory by H\'ajek and Le Cam that generalize the
Cram\'er-Rao bound but are much more sophisticated
\cite{gill95,vaart}.  Equation~(\ref{natural}) is not the only
contravariant choice, however. It does not even exist if $u$ is not in
the range of the $F$ matrix \cite{stoica01}. It is also not the
optimal choice for the Gill-Levit bounds in general, as
Sec.~\ref{sec_optimal} later shows.

Another useful choice of the $v$ components is
\begin{align}
v^a &= \Bk{\bk{n\Avg{F} + \Avg{G}}^{-1}}^{ab} \Avg{u_b},
\label{vt}
\\
G_{ab} &\equiv 
\frac{1}{\pi}\bk{\partial_a \pi}\frac{1}{\pi}\bk{\partial_b \pi},
\end{align}
leading to
\begin{align}
\mathsf B &= 
\Avg{u_a} \Bk{\bk{n\Avg{F} + \Avg{G}}^{-1}}^{ab} \Avg{u_b}.
\label{B_original}
\end{align}
If $u$ is $\theta$-independent, Eq.~(\ref{B_original}) coincides with
the original version by Sch\"utzenberger and Van Trees
\cite{schutzenberger57,vantrees}.  $\avg{G}$ plays the role of prior
information and can regularize the inverse when $\avg{F}$ is
ill-conditioned. The regularization is especially important for
waveform-estimation problems \cite{vantrees,twc}. The form of
  Eq.~(\ref{vt}) is usually not contravariant, however, in the sense
  that, except for special cases, Eq.~(\ref{vt}) for one
  parametrization and
  $\tilde v^a = [(n\avg{\tilde F} + \avg{\tilde G})^{-1}]^{ab}
  \avg{\tilde u_b}$ for another parametrization do not obey
  Eq.~(\ref{contra}), and the resultant bounds may be different for a
  given problem.

In the following, I generalize $\Theta$, the parameter space, to
  a $p$-dimensional manifold, and assume that $v$ is a vector field on
  the manifold. The formalism can then be made more elegant by
defining the invariant quantities
\begin{align}
\epsilon &\equiv \sqrt{|g|} d^p\theta,
&
\prior &\equiv \frac{\pi}{\sqrt{|g|}},
&
\pi d^p\theta &= \prior \epsilon,
\end{align}
where $|g|$ is the determinant of a Riemannian (positive-definite)
metric $g_{ab}$. It should be emphasized that the metric here is
merely a mathematical tool to keep track of parametrization invariance
and deal with more general manifolds for $\Theta$, and this work is
not concerned with the concept of statistical manifolds and distances
between probability measures in information geometry
\cite{amari}. Although many have argued that the Fisher information is
a natural metric in information geometry \cite{amari}, there is no
particular reason to pick the Fisher information as the metric
here. That choice may also cause problems if $|F| = 0$, so I keep the
metric unspecified here for generality.  The divergence term in
Eq.~(\ref{P}) becomes
\begin{align}
\frac{1}{\pi}\partial_a\bk{\pi v^a} 
&= \frac{1}{\sqrt{|g|}\prior} \partial_a\bk{\sqrt{|g|} \prior v^a}
= \frac{1}{\prior} \nabla_a\bk{\prior v^a},
\end{align}
where $\nabla_a$ is the Riemannian covariant derivative. With these
suggestive expressions at hand, I propose the following.

\begin{proposition}[Invariant Gill-Levit bounds]
  If $\prior v$ vanishes on any boundary of the parameter manifold
  $\Theta$, the Bayesian mean-square risk has a lower bound given by
  Eq.~(\ref{GL}), where
\begin{align}
\Avg{\mathsf A} &= \int \bk{v^a u_a}  \prior \epsilon,
\label{invariant_N}
\\
\Avg{\mathsf F} &= \int \bk{v^a F_{ab} v^b} \prior \epsilon,
\label{invariant_F}
\\
\Avg{\mathsf P} &= 
\int  \Bk{\frac{1}{\prior}\nabla_a\bk{\prior v^a}}^2 \prior\epsilon.
\label{invariant_P}
\end{align}
\label{prop_invariant}
\end{proposition} 
\begin{proof}
  For completeness, I provide a proof that proceeds in a manifestly
  invariant way, so that the proposition is proved also for a curved
  metric. Define the bias as
\begin{align}
\mathsf b &\equiv \int \bk{\check\beta - \beta} f^{(n)} d\mu^{(n)},
\end{align}
and write, via the Leibniz rule for the covariant derivative,
\begin{align}
\int \nabla_a\bk{\mathsf b \prior v^a}\epsilon
&= \iint \bk{\check\beta-\beta} \nabla_a\bk{f^{(n)}\prior v^a} 
d\mu^{(n)}\epsilon
\nonumber\\
&\quad -\int \bk{v^a\nabla_a\beta} \prior \epsilon.
\label{leibniz}
\end{align}
It can be shown that the left-hand side of Eq.~(\ref{leibniz}) is zero
by applying the Stokes theorem \cite{carroll,lee03} and requiring that
$\prior v$ vanishes on the boundary of $\Theta$ if there is a
boundary.  With $\nabla_a\beta = \partial_a\beta$ when $\nabla_a$ acts
on a scalar, the last term in Eq.~(\ref{leibniz}) is precisely
$\avg{\mathsf A}$ in Eq.~(\ref{invariant_N}). I obtain
\begin{align}
\Avg{\mathsf A} &= 
\iint \bk{\check\beta-\beta} \nabla_a\bk{f^{(n)}\prior v^a} d\mu^{(n)}\epsilon
\\
&= \expect\Bk{\bk{\check\beta-\beta}s},
\end{align}
where $s$ is a generalized score function given by
\begin{align}
s &\equiv \frac{1}{f^{(n)}\prior} \nabla_a\bk{f^{(n)}\prior v^a}
\\
&=\frac{1}{f^{(n)}} v^a\nabla_a  f^{(n)} 
+ \frac{1}{\prior}\nabla_a\bk{\prior v^a}.
\end{align}
The expectation can be regarded as an inner product. The
Cauchy-Schwarz inequality then gives
\begin{align}
\Avg{\mathsf A}^2 &\le 
\expect\Bk{\bk{\check\beta-\beta}^2}
\expect\bk{s^2}.
\label{cauchy}
\end{align}
With the usual premise
\begin{align}
\int \nabla_a f  d\mu
&= \int \partial_a  f(x|\theta) d\mu(x) = 
\partial_a\int f d\mu = 0,
\end{align}
it can be shown that
\begin{align}
\expect\bk{s^2}
= n \Avg{\mathsf F} + \Avg{\mathsf P},
\end{align}
with $\avg{\mathsf F}$ given by Eq.~(\ref{invariant_F}) and
$\avg{\mathsf P}$ given by Eq.~(\ref{invariant_P}).  Hence,
Eq.~(\ref{cauchy}) leads to Eq.~(\ref{GL}), together with
Eqs.~(\ref{invariant_N})--(\ref{invariant_P}).
\end{proof}

The original Gill-Levit bounds given by Eqs.~(\ref{GL})--(\ref{P}) may
be viewed as a special case of Proposition~\ref{prop_invariant} if one
can pick a parametrization (coordinate system) with
$g_{ab}=\delta_{ab}$ everywhere in $\Theta$. If the Riemann curvature
tensor with respect to the metric is zero everywhere, then one can
always find a parametrization for which $g_{ab} = \delta_{ab}$
\cite{carroll}, and the two formulations are equivalent in
essence. But if not, the metric is said to be curved, and
Proposition~\ref{prop_invariant} is more
general. Proposition~\ref{prop_invariant} may also be regarded as a
special case of Theorem~2.1 in Ref.~\cite{jupp10}, although the latter
is so general that the bound there may depend on the estimator.

While it is unclear whether curved metrics are useful for the kind of
problems considered here, one immediate advantage of the invariant
formulation is that all the ensuing results are guaranteed to be
invariant.

\section{\label{sec_optimal}Optimal Gill-Levit bound}

To derive the optimal Gill-Levit bound, it is illuminating to first
recall the concept of least favorable submodels in the local theory,
as outlined in Ref.~\cite{stein56}; see also Ref.~\cite{gross20}. Pick
a curve in the parameter space that passes through the true value and
denote a tangent vector there as $v$. The local bound for the
one-dimensional submodel is given by
\begin{align}
\mathsf C(v)&= \frac{(v^a u_a)^2}{v^a F_{ab} v^b}.
\end{align}
Define an inner product between two vectors as
\begin{align}
\Avg{v,w}_g &\equiv  v^aw_a = v^a g_{ab} w^a,
\end{align}
where the usual convention of index lowering and raising via $g_{ab}$
and its inverse $g^{ab}$ in differential geometry is assumed.  Let $F$ be an operator that obeys
$(F v)_a = F_{ab}v^b$. If $F$ is positive-definite, $F^{-1}$ and the
square roots $F^{1/2}$ and $F^{-1/2}$ exist \cite{horn}.  The
Cauchy-Schwarz inequality gives
\begin{align}
\mathsf C(v) &= \frac{\avg{v,u}_g^2}{\avg{v,Fv}_g}
= \frac{\avg{F^{1/2}v,F^{-1/2}u}_g^2}{\avg{v,Fv}_g} \\
&\le \Avg{u,F^{-1}u}_g = u_a\bk{F^{-1}}^{ab}u_b,
\end{align}
which coincides with Eq.~(\ref{CRB}) for the full model. A least
favorable tangent vector that attains the equality must satisfy
\begin{align}
v^a &\propto \bk{F^{-1}}^{ab} u_b.
\label{least_local}
\end{align}
Thus, Eq.~(\ref{CRB}) can be evaluated by considering the tangent
space at the true parameter and picking the worst direction.

For the Gill-Levit bounds, the ``natural'' choice of $v$ given by
Eq.~(\ref{natural}) is a least favorable choice in the local
theory. Thus, one may intuit that $v$ plays an analogous role of
picking out directions in the Bayesian bound, except that $v$ should
now be considered as a vector field, as depicted in
Fig.~\ref{vector_field}.  In differential geometry, a vector field can
generate a family of integral curves, called a flow, in the manifold,
and vice versa \cite{lee03}. In the context of statistics, each curve
corresponds to a one-dimensional submodel, so the concept of locally
least favorable submodels may be generalized to a concept of least
favorable flows.  Following this intuition, I can generalize the
strategy of optimizing over $v$ to obtain the tightest bound, as
follows.

\begin{figure}[htbp!]
\centerline{\includegraphics[width=0.45\textwidth]{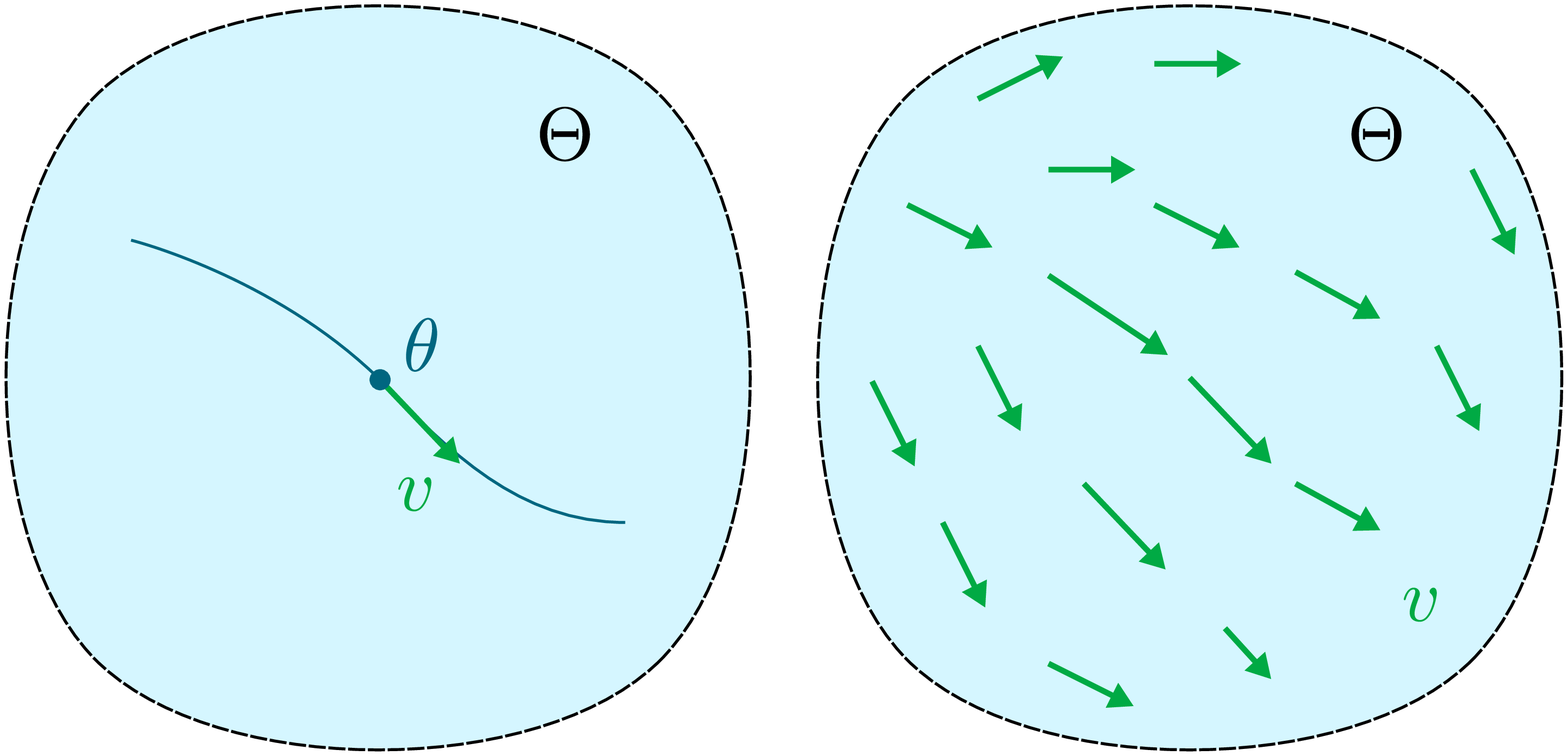}}
\caption{\label{vector_field}Left: a geometric picture of a
  one-dimensional submodel as a curve in the manifold and a tangent
  vector $v$ at the true parameter value $\theta$ in the local theory.
  Right: a picture of $v$ as a vector field in the Bayesian theory.}
\end{figure}

\begin{theorem}[Optimal Gill-Levit bound]
\label{thm_optimal}
\begin{align}
\max_v \mathsf B = \Avg{u,L^{-1}u}_\prior \equiv \mathsf B_{\rm max},
\label{Bmax}
\end{align}
where the inner product between two vector fields is defined as
\begin{align}
\Avg{v,u}_\prior &\equiv \int v^au_a \prior \epsilon,
\label{inner_rho}
\end{align}
the linear, self-adjoint, and positive-semidefinite operator $L$ is
defined as
\begin{align}
(Lv)_a &\equiv
 n F_{ab} v^b -\nabla_a\Bk{\frac{1}{\prior}\nabla_b\bk{\prior v^b}},
\label{L}
\end{align}
and $u$ is assumed to be in the range of $L$, such that $L^{-1}u$
exists. A least favorable vector field, defined as a $v$ that
maximizes $\mathsf B$, must satisfy
\begin{align}
v &\propto L^{-1} u.
\label{least_global}
\end{align}
\end{theorem}
\begin{proof}
  In terms of the inner product given by Eq.~(\ref{inner_rho}),
  Eqs.~(\ref{invariant_N})--(\ref{invariant_P}) can be expressed as
\begin{align}
\Avg{\mathsf A} &= \Avg{v,u}_\prior,
\\
\Avg{\mathsf F} &= \Avg{v,Fv}_\prior,
\\
\Avg{\mathsf P} &= 
\int \Bk{\nabla_a\bk{\prior v^a}}
 \frac{1}{\prior}\nabla_b\bk{\prior v^b} \epsilon
\\
&= 
-\int \prior v^a \nabla_a\Bk{\frac{1}{\prior}\nabla_b\bk{\prior v^b}} \epsilon
\label{parts}
\\
&= \Avg{v,Pv}_\prior,
\\
(Pv)_a &\equiv  
-\nabla_a\Bk{\frac{1}{\prior}\nabla_b\bk{\prior v^b}},
\end{align}
where Eq.~(\ref{parts}) comes from integration by parts, as enabled by
the Leibniz rule and the Stokes theorem, and the assumption that
$\prior v$ vanishes on any boundary of $\Theta$. One can check that
$F$ and $P$ are linear, self-adjoint, and positive-semidefinite
operators. Furthermore,
\begin{align}
n \Avg{\mathsf F} + \Avg{\mathsf P} &= \Avg{v,Lv}_\prior,
&
L &= n F + P.
\end{align}
As $L^{-1}u$ is assumed to exist, the Cauchy-Schwarz inequality
yields
\begin{align}
\mathsf B &= \frac{\Avg{v,u}_\prior^2}{\Avg{v,Lv}_\prior}
= \frac{\avg{L^{1/2} v,L^{-1/2} u}_\prior^2}{\avg{v,Lv}_\prior}
\le \Avg{u,L^{-1}u}_\prior,
\end{align}
and the equality is attained if and only if $v$ obeys
Eq.~(\ref{least_global}).
\end{proof}

Within the Gill-Levit family, $\mathsf B_{\rm max}$ is not only the
maximum but also the closest in spirit to the local bound given by
Eq.~(\ref{CRB}), with the $L^{-1}$ operator generalizing the role of
$F^{-1}$. Moreover, note that $\mathsf B_{\rm max}$ is naturally
invariant. Although it is also possible to derive
$\mathsf B_{\rm max}$ starting from Eqs.~(\ref{GL})--(\ref{P}) without
the invariant formalism, at least for a flat metric, the invariance of
$\mathsf B_{\rm max}$ would have been much more tedious to prove, with
a proliferation of Jacobians.

The most difficult part of computing $\mathsf B_{\rm max}$ is solving
for $L^{-1}u$. Let $v = L^{-1}u$, which is a least favorable field. It
obeys the second-order field equation
\begin{align}
\bk{Lv}_a &= nF_{ab} v^b -\nabla_a\Bk{\frac{1}{\prior}\nabla_b\bk{\prior v^b}} 
= u_a.
\label{least_explicit}
\end{align}
The solution, expressible in terms of an impulse-response (Green)
function, can be substituted into Eq.~(\ref{Bmax}) to give
$\mathsf B_{\rm max}$. For large $n$, Eq.~(\ref{least_explicit}) can
be simplified to
\begin{align}
n F_{ab} v^b &\approx u_a,
&
\mathsf B_{\rm max} &\approx \frac{\avg{\mathsf C}}{n},
\end{align}
so Eq.~(\ref{least_local}) is asymptotically least favorable to the
Gill-Levit family, in nice agreement with the local theory
\cite{vaart} and earlier results \cite{abushanab15,*koike20}. Note,
however, that the exact optimal choice according to
Eq.~(\ref{least_explicit}) also depends on the prior and some
derivatives. The correction to the local theory becomes especially
important if $u$ is not in the range of $F$ and
Eq.~(\ref{least_local}) has no solution. The question of what to do
when $u$ is not even in the range of $L$, and
Eq.~(\ref{least_explicit}) has no solution, remains open.

Another special case is when a parametrization with
$g_{ab} = \delta_{ab}$ is assumed, $u$ and $F$ are
$\theta$-independent, and $\pi$ is Gaussian with covariance matrix
$G^{-1}$. Then the solution to Eq.~(\ref{least_explicit}) is
\begin{align}
v^a &= \Bk{\bk{nF + G}^{-1}}^{ab}u_b,
\end{align}
and $\mathsf B_{\rm max}$ becomes
\begin{align}
\mathsf B_{\rm max} &= 
u_a \Bk{\bk{nF + G}^{-1}}^{ab}u_b,
\label{Bmax_gauss}
\end{align}
which coincides with the Sch\"utzenberger-Van Trees version given by
Eq.~(\ref{B_original}), since $\avg{u} = u$, $\avg{F} = F$, and
$\avg{G} = G$ in this case. Furthermore, if
  $f(x|\theta) = f(x-\theta)$ and $f(x-\theta)$ is also Gaussian, such
  that $F$ is the inverse of the covariance matrix of $f$, then it is
  well known that the minimum Bayes risk
  $\min_{\check\beta}\avg{\mathsf R}$ is also given by the right-hand
  side of Eq.~(\ref{Bmax_gauss}) \cite{vantrees}, and
  $\mathsf B_{\rm max}$ is a tight bound.

\section{\label{sec_wave}Wave picture}
I now switch gears and make the substitution
\begin{align}
\prior &= \psi^2,
\end{align}
where $\psi$ is a real function of the parameter. I call $\psi$ a
wavefunction. All the functionals in
Eqs.~(\ref{invariant_N})--(\ref{invariant_P}) turn out to be quadratic
with respect to $\psi$ and $\nabla_a\psi$, given by
\begin{align}
\Avg{\mathsf A} &= \int \bk{v^a u_a}\psi^2  \epsilon,
\label{A_wave}
\\
\Avg{\mathsf F} &= \int  \bk{v^a F_{ab} v^b}\psi^2\epsilon,
\\
\Avg{\mathsf P} &= \int \bk{\mathsf D\psi}^2\epsilon,
\\
\mathsf D\psi &\equiv \bk{\nabla_a v^a}\psi + 2 v^a\nabla_a\psi.
\label{D}
\end{align}
The problem of choosing an unfavorable prior to tighten the bound for
minimax estimation now becomes a problem of finding the wavefunction
that maximizes $\mathsf B$. To simplify, I define yet another inner
product as
\begin{align}
\Avg{\psi,\phi} &\equiv \int \psi \phi\epsilon.
\end{align}
The normalization condition for the prior density becomes
\begin{align}
\int \prior \epsilon = \Avg{\psi,\psi} = 1.
\label{normalization}
\end{align}
It can be shown that
\begin{align}
\Avg{\mathsf A} &= \Avg{\psi,\mathsf A\psi},
\\
\Avg{\mathsf F} &= \Avg{\psi,\mathsf F\psi},
\\
\Avg{\mathsf P} &= \Avg{\mathsf D\psi,\mathsf D\psi} 
= \Avg{\psi,\mathsf D^\dagger \mathsf D\psi},
\label{P_wave2}
\\
\mathsf D^\dagger\psi &= \bk{\nabla_a v^a}\psi - 2 v^a\nabla_a\psi,
\label{D_dagger}
\\
\mathsf B &= \frac{\Avg{\psi,\mathsf A\psi}^2}{\Avg{\psi,\mathsf H\psi}},
\\
\mathsf H &\equiv n \mathsf F + \mathsf D^\dagger \mathsf D.
\label{H}
\end{align}
Note that $\mathsf D$ may be a nonlinear operator, if the choice of
$v$, such as Eq.~(\ref{least_explicit}), depends on the prior.  To
proceed, I assume that $v$ does not depend on $\psi$ and $\mathsf D$
is linear.  Then I can follow the approach in Sec.~\ref{sec_optimal}
to obtain
\begin{align}
\mathsf B &= 
\frac{\Avg{\mathsf H^{1/2}\psi, \mathsf H^{-1/2}\mathsf A\psi}^2}
{\Avg{\psi,\mathsf H\psi}}
\le \Avg{\mathsf A\psi, \mathsf H^{-1}\mathsf A\psi}.
\end{align}
The equality is attained if and only if
\begin{align}
\mathsf H \psi &= \bk{n \mathsf F + \mathsf D^\dagger \mathsf D}\psi 
= \lambda\mathsf A\psi,
\label{wave}
\end{align}
where $\lambda$ is an arbitrary nonzero real number.  Let
$\psi_\lambda$ be a solution of Eq.~(\ref{wave}) as a function of
$\lambda$, subject to the normalization constraint given by
Eq.~(\ref{normalization}). Then
\begin{align}
\mathsf B &= 
\frac{1}{\lambda}\Avg{\psi_\lambda,\mathsf A\psi_\lambda},
\end{align}
and this expression should be maximized with respect to $\lambda$ to
obtain the tightest lower bound on $\sup_\theta \mathsf R(\theta)$.

A substantial simplification can be made if a parametrization with
$g_{ab}=\delta_{ab}$ can be assumed and $u$, $v$, and therefore
$\mathsf A$ are $\theta$-independent. Equation~(\ref{wave}) becomes
\begin{align}
\Bk{n\mathsf F(\theta) -4\bk{v^a\partial_a}^2}\psi(\theta) 
&= \lambda \mathsf A \psi(\theta),
\label{schrodinger}
\end{align}
which is a time-independent Schr\"odinger equation. The Fisher
information $\mathsf F = v^aF_{ab} v^b$, evaluated in the direction of
$v$, plays the role of the potential, while $-(v^a\partial_a)^2$, in
terms of the directional derivative $v^a\partial_a$, plays the role of
the kinetic-energy operator. The bound becomes
\begin{align}
\mathsf B &= \frac{\mathsf A^2}{\Avg{\psi,\mathsf H\psi}}.
\label{B_average_energy}
\end{align}
To maximize $\mathsf B$, one should therefore solve for
\begin{align}
\mathsf B_{\rm worst} &\equiv \sup_{\psi:\avg{\psi,\psi}=1}\mathsf B
= \frac{\mathsf A^2}{\mathsf E_{\rm min}},
\label{Bworst}
\\
\mathsf E_{\rm min} &\equiv 
\inf_{\psi: \avg{\psi,\psi} = 1}\Avg{\psi,\mathsf H\psi},
\label{Emin}
\end{align}
that is, the ground-state energy.  The infimum is used here in
  case a normalizable ground state does not exist.  Adding a phase to
the wavefunction cannot reduce the energy, so the consideration of
only real wavefunctions is justified here.

The wave correspondence makes sense, as intuition suggests that an
unfavorable prior should be concentrated near the minimum of the
Fisher information, just as the ground state should be concentrated
near the bottom of the potential. If the prior density is made too
sharp, however, the prior information $\avg{\mathsf P}$ would become
large, and therefore a balance between $n\avg{\mathsf F}$ and
$\avg{\mathsf P}$ should be struck to minimize their sum, just as the
ground state achieves the optimal balance between the potential and
kinetic energies.

In the limit $n \to \infty$, the ground-state energy is the
classical-mechanics limit given by
\begin{align}
\mathsf E_{\rm min} &= n \inf_{\theta\in\Theta} \mathsf F(\theta) + o(n),
\end{align}
where $o(g(n))$ denotes a term in a smaller order than $g(n)$ as
$n\to\infty$. Other asymptotic notations \cite{knuth76},
  including $\mathit\Theta(g(n))$ (same order as $g(n)$) and
  $\mathit\Omega(g(n))$ (order at least as large as $g(n)$), will also be
  used in the following. If the infimum of $\mathsf F(\theta)$ is
strictly positive, $\mathsf B_{\rm worst}$ obeys the parametric rate
$\mathit\Theta(n^{-1})$. A more interesting case is when the infimum
is zero, $\mathsf E_{\rm min} = o(n)$, and the bound mandates a
convergence rate slower than the parametric rate. A concrete special
case is as follows.
\begin{theorem}
\label{thm_rate}
  Suppose that $u$ and $v$ are $\theta$-independent and obey
  $v^a u_a \neq 0$. Suppose also that there exists a one-dimensional
  submodel with parametrization
\begin{align}
\theta^a(\tau) = \theta^a(0) + v^a\tau,
\end{align}
$\tau \in (\tau_1,\tau_2) \subseteq \mathbb R$,
$\tau_1 \le 0 \le \tau_2$, $\tau_1\neq\tau_2$, and Fisher information
bounded by
\begin{align}
\mathsf F(\tau) &= v^a F_{ab}\bk{\theta(\tau)}v^b
\le A \abs{\tau}^m,
\end{align}
where $A$ is a positive constant and $m \ge 0$. Then
\begin{align}
\sup_{\theta\in\Theta}\mathsf R(\theta)
&\ge \mathsf B_{\rm worst} = \mathit\Omega\bk{n^{-2/(m+2)}}.
\end{align}
\end{theorem}
\begin{proof}
With the given conditions, the average energy for the submodel is
\begin{align}
\Avg{\psi,\mathsf H\psi}
=  \int_{\tau_1}^{\tau_2} 
\BK{n\mathsf F(\tau) \psi(\tau)^2 + \Bk{\parti{\psi(\tau)}{\tau}}^2} d\tau.
\end{align}
Let $\psi(\tau) = \phi(\tau/W)/\sqrt{W}$, where $\phi$ is a trial
function and $0 < W \le 1$ scales the width of $\psi$. Then
\begin{align}
\Avg{\psi,\mathsf H\psi}
&\le n A W^m \int_{\tau_1/W}^{\tau_2/W}
\phi(y)^2 \abs{y}^m dy
\nonumber\\&\quad
+ \frac{4}{W^2}\int_{\tau_1/W}^{\tau_2/W} \Bk{\parti{\phi(y)}{y}}^2 dy
\\
&= 
n A W^m \int_{\tau_1}^{\tau_2}
\phi(y)^2 \abs{y}^m dy
\nonumber\\&\quad
+ \frac{4}{W^2}\int_{\tau_1}^{\tau_2} \Bk{\parti{\phi(y)}{y}}^2 dy,
\label{variational}
\end{align}
where the last step uses the fact that $\tau_1/W \le \tau_1$ and
$\tau_2/W \ge \tau_2$, since $\tau_1 \le 0 \le \tau_2$ and
$0 < W \le 1$, and $\phi(y)$ vanishes outside $(\tau_1,\tau_2)$.  It
is not difficult to show that, regardless of $\tau_1$ and $\tau_2$,
there always exists a trial function that makes both integrals in
Eq.~(\ref{variational}) converge. Minimizing Eq.~(\ref{variational})
with respect to $W$, I obtain
\begin{align}
  W &=  A_1 n^{-1/(m+2)},
\\
  \Avg{\psi,\mathsf H\psi} &\le A_2 n^{2/(m+2)},
\label{average_H}
\end{align}
where $A_1$ and $A_2$ are positive constants.  For a large enough $n$,
the assumption $W \le 1$ can be satisfied. The theorem then follows from
Eqs.~(\ref{GL}), (\ref{minimax}), (\ref{Bworst}), (\ref{Emin}), and
(\ref{average_H}).
\end{proof}

A concrete example is $\mathsf F(\tau) \le A\tau^2$, in which
  case we can borrow from the theory of quantum harmonic oscillators
  to find that the ground-state energy for a potential $n A \tau^2$ is
  $\mathit\Theta(n^{1/2})$, leading to
  $\mathsf B_{\rm worst} = \mathit\Omega(n^{-1/2})$.

\section{\label{sec_quantum}Quantum estimation theory}
\subsection{Basics}
Assume $n = 1$ without loss of generality. Let
$\{\dop(\theta):\theta \in \Theta\}$ be a family of density operators
that model a quantum system. The generalized Born's rule states that
the statistics of any measurement of the system can be modeled by a
positive operator-valued measure (POVM) $E$ \cite{hayashi} via
\begin{align}
f(x|\theta) d\mu(x) &= \trace\Bk{dE(x) \dop(\theta)},
\end{align}
where $\trace$ denotes the operator trace. For any POVM, an upper
bound on the Fisher information is given by \cite{young75,*nagaoka87,hayashi}
\begin{align}
\mathsf F &= v^a F_{ab} v^b \le v^a K_{ab} v^b \equiv \mathsf K
\label{nagaoka}
\end{align}
for any vector $v$, where $K$ is the Helstrom information matrix
\cite{helstrom} defined as
\begin{align}
K_{ab}(\theta) &\equiv \trace\Bk{ \dop(\theta) 
{\mathcal S}_a(\theta) \circ {\mathcal S}_b(\theta)},
\end{align}
$A\circ B \equiv (AB+BA)/2$ denotes the Jordan product, and
${\mathcal S}_a$, a score operator, is a solution to
\begin{align}
\partial_a \dop(\theta) &= \dop(\theta) \circ 
{\mathcal S}_a(\theta).
\end{align}
There exist other quantum versions of the Fisher information and the
Cram\'er-Rao bound that are of interest when $\beta$ is vectoral
\cite{hayashi,gill_guta,demkowicz20,suzuki20}, but they are
  outside the scope of this work, and I focus on the Helstrom
information hereafter.

With Eq.~(\ref{nagaoka}), a quantum lower bound on $\mathsf B$ for any
POVM can be obtained simply by replacing $F$ with $K$. To be explicit,
\begin{align}
\Avg{\mathsf R} \ge
\mathsf B \ge \mathsf Q \equiv \frac{\avg{\mathsf A}^2}
{\avg{\mathsf K} + \avg{\mathsf P}}.
\end{align}
For $\mathsf B$ to attain $\mathsf Q$, the equality in
Eq.~(\ref{nagaoka}) must hold for all $\theta \in \Theta$, and that is
usually not possible.

As $K$ is also a positive-semidefinite $(0,2)$ tensor, all the results
in the previous sections apply to the quantum bound as well. In
particular, following Theorem~\ref{thm_optimal}, the optimal
$\mathsf Q$ is
\begin{align}
\mathsf Q_{\rm max} &\equiv \max_v \mathsf Q
= \Avg{u,R^{-1}u}_\prior,
\label{Qmax}
\\
\bk{R v}_a &\equiv K_{ab} v^b - 
\nabla_a\Bk{\frac{1}{\prior}\nabla_b\bk{\prior v^b}}.
\label{R}
\end{align}
It is not difficult to prove that
\begin{align}
\mathsf Q_{\rm max} &\le \mathsf B_{\rm max}
\end{align}
for any POVM.

A simple example is the quantum Gaussian shift model, where
$\dop(\theta)$ is the quantum state of $m$ harmonic oscillators with a
Gaussian Wigner representation and $\theta \in \mathbb R^p$, with
$p = 2m$, is its displacement in phase space
\cite{holevo11,demkowicz20}. Assuming the standard parametrization
with $g_{ab} = \delta_{ab}$, $K$ is the inverse of the covariance
matrix of the Wigner function and $\theta$-independent. By measuring
the object together with an auxiliary in a Gaussian state with the
same covariance matrix, it is possible to produce classical Gaussian
shift statistics that achieves $F = K/2$ \cite{albarelli20}. If $\pi$
is also Gaussian with covariance matrix $G^{-1}$ and $u$ is
$\theta$-independent, then, by the same rationale that gives
Eq.~(\ref{Bmax_gauss}), it can be shown that
\begin{align}
\mathsf Q_{\rm max} &= u^\top \bk{K+G}^{-1} u,
\end{align}
and for the measurement just mentioned,
\begin{align}
\min_{\check\beta}\Avg{\mathsf R} &= \mathsf B_{\rm max}
= u^\top \bk{K/2+G}^{-1} u,
\\
\mathsf Q_{\rm max} &\le 
\min_{\check\beta}\Avg{\mathsf R} \le 2\mathsf Q_{\rm max}.
\end{align}
A further optimization of the measurement for a given $u$ may be
possible, but the optimization problem becomes more difficult in
general, especially when $u$ is $\theta$-dependent or $\pi$ is
non-Gaussian.

\subsection{\label{sec_waveform}Waveform estimation}
Consider a quantum dynamical system, such as the optomechanical force
sensor depicted in Fig.~\ref{optomech}, under the influence of a
classical waveform $\theta(t)$. Using the principles of larger Hilbert
space and deferred measurements \cite{nielsen}, the statistics of a
sequentially measured quantum system can be modeled by a POVM at the
final time and a density-operator family given by
\begin{align}
\dop(\theta) &= U(\theta) \ket{\Psi}\bra{\Psi} U(\theta)^\dagger,
\\
U(\theta) &= \mathcal T\exp\BK{\frac{1}{i\hbar} \int_{-T/2}^{T/2}
\Bk{H_0(t)-q \theta(t)} dt},
\end{align}
where $\ket{\Psi}$ is the initial state of the quantum system, $q$ is
a position operator, $H_0(t)$ is the rest of the Hamiltonian, $T$ is
the total observation time, and $\mathcal T$ denotes time ordering of
the operator exponential.

\begin{figure}[htbp!]
\centerline{\includegraphics[width=0.48\textwidth]{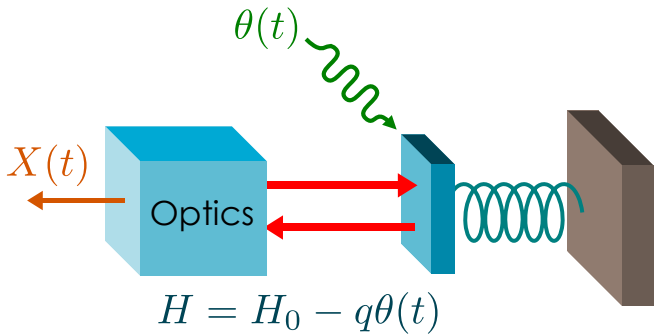}}
\caption{\label{optomech}An optomechanical force sensor under
  continuous optical measurements. $\theta(t)$ is the unknown
  classical force, $H$ is the system Hamiltonian, and $X(t)$ is the
  observation process.}
\end{figure}

Let the parameter of interest be defined in
terms of a weight function $h(t)$ as
\begin{align}
\beta &= \int_{-T/2}^{T/2} h(t)\theta(t) dt.
\end{align}
For example, if $\beta = \theta(\tau)$ at an instant of time $\tau$ is
of interest, then $h(t) = \delta(t-\tau)$. To derive analytic
results, I follow Ref.~\cite{twc} and discretize time as
\begin{align}
t_a &= -\frac{T}{2}+a\delta t,
&
T &= p \delta t.
\end{align}
Assuming 
\begin{align}
\theta(t_a) &= \theta^a, & h(t_a) &= h_a,
\end{align}
and
\begin{align}
\beta &\approx h_a\theta^a\delta t,
\\
U \approx  U(t_p,t_1)&\equiv \exp\Bk{\frac{H_0(t_p) \delta t}{i\hbar}}
\exp\bk{\frac{iq\theta^p\delta t}{\hbar}}\dots
\nonumber\\&\quad
\exp\Bk{\frac{H_0(t_1) \delta t}{i\hbar}}
\exp\bk{\frac{iq\theta^1\delta t}{\hbar}},
\end{align}
it can be shown that
\begin{align}
u_a &\approx 
\partial_a\bk{h_a\theta^a\delta t} =  h_a \delta t,
\label{ha}
\\
K_{ab} &\approx \frac{4\delta t^2}{\hbar^2}C_q(t_a,t_b),
\label{K_waveform}
\\
C_q(t_a,t_b)
&\equiv 
\bra{\Psi}\hat q(t_a)\circ \hat q(t_b)\ket{\Psi}
\nonumber\\&\quad
-\bra{\Psi}\hat q(t_a)\ket{\Psi}\bra{\Psi}\hat q(t_b)\ket{\Psi},
\end{align}
where 
\begin{align}
\hat q(t_a) &\equiv U(t_{a-1},t_1)^\dagger q U(t_{a-1},t_1)
\end{align}
is the Heisenberg picture of $q$, $C_q$ is its covariance function,
and the right-hand side of Eq.~(\ref{K_waveform}) is the exact
  Helstrom information for
  $\dop(\theta) =
  U(t_p,t_1)\ket{\Psi}\bra{\Psi}U(t_p,t_1)^\dagger$. If $\hat q(t)$
is stationary, the covariance can be written in terms of a power
spectral density $S_q(\omega)$ as \cite{braginsky}
\begin{align}
C_q(t_a,t_b) &= 
\intall S_q(\omega) \exp\Bk{i\omega(t_b-t_a)}
\frac{d\omega}{2\pi}.
\end{align}
With the assumption of stationary processes and long observation time
(SPLOT) \cite{vantrees}, $K$ can be approximated as a circulant matrix
\cite{gray06} and expressed as
\begin{align}
K_{ab} &\approx 
\frac{\delta t}{p}\sum_{j=0}^{p-1} \frac{4S_q(\omega_j)}{\hbar^2}
\exp\Bk{i\omega_j(t_b-t_a)},
\label{K_stationary}
\end{align}
where $\omega_j = \omega_0 + 2\pi j/T$ and $\omega_0 = -\pi/\delta t$.
Similarly, if $\theta(t)$ is a stationary Gaussian random process with
power spectral density $S_\theta(\omega)$,
\begin{align}
G_{ab} &\approx \frac{\delta t}{p}
\sum_{j=0}^{p-1} \frac{1}{S_\theta(\omega_j)}\exp\Bk{i\omega_j(t_b-t_a)}.
\end{align}
As $V_{ja} \equiv \exp(-i\omega_jt_a)/\sqrt{p}$ is a unitary matrix,
the inverse of $K + G$ can be computed analytically to give
\begin{align}
\Bk{\bk{K+G}^{-1}}^{ab} &\approx \frac{1}{T}
\sum_{j=0}^{p-1} \frac{\exp\Bk{i\omega_j\bk{t_a-t_b}}}
{4S_q(\omega_j)/\hbar^2 + 1/S_\theta(\omega_j)}.
\end{align}
$u$, as given by Eq.~(\ref{ha}), does not depend on $\theta$.  If the
dynamics of the system is linear \cite{braginsky}, $K$ also does not
depend on $\theta$.  Thus, the same argument that leads to
Eq.~(\ref{Bmax_gauss}) can be used to give
\begin{align}
\mathsf Q_{\rm max} &= u_a\Bk{\bk{K+G}^{-1}}^{ab}u_b,
\\
&\approx
\frac{1}{T} \sum_{j=0}^{p-1} 
\frac{\delta t^2 h_a h_b\exp\Bk{i\omega_j\bk{t_a-t_b}}}
{4S_q(\omega_j)/\hbar^2+1/S_\theta(\omega_j)}.
\end{align}
Taking the continuous and long time limit with $\delta t \to 0$,
$T\to\infty$, and $d\omega = 2\pi/T$ hence results in
\begin{align}
\mathsf Q_{\rm max} &\to \intall
\frac{|\tilde h(\omega)|^2}
{4S_q(\omega)/\hbar^2+1/S_\theta(\omega)}\frac{d\omega}{2\pi},
\label{Qmax_waveform}
\\
\tilde h(\omega) &\equiv 
\intall h(t)\exp(-i\omega t)dt.
\end{align}
If $\beta = \theta(\tau)$ with $h(t) = \delta(t-\tau)$ and
$|\tilde h(\omega)| = 1$, Eq.~(\ref{Qmax_waveform}) agrees with the
result in Ref.~\cite{twc}.  Compared with Ref.~\cite{twc}, which
derives a quantum bound on $\avg{\mathsf R}$ directly, the derivation
here clarifies the relation of Eq.~(\ref{Qmax_waveform}) to the
Helstrom information and the Gill-Levit formalism. The new insight
implied by the theory here is that the bound remains invariant upon
any reparametrization and cannot be further improved by picking a
different $v$.

While Eq.~(\ref{Qmax_waveform}) holds for any measurement, it can say
something more about measurements in the linear form of
\begin{align}
X(t) &= \intall h_X(t-t')\theta(t')dt' + Z(t),
\label{X_process}
\end{align}
where $h_X$ is an impulse-response function of the system and $Z$
is a stationary noise process that is uncorrelated with $\theta$. In
optomechanics, such a process can be obtained by homodyne detection of
the output light. Let the estimator be
\begin{align}
\check\beta &= \intall \check h(t) X(t) dt,
\end{align}
where $\check h(t)$ is a linear filter, or more precisely a smoother
in control-theoretic terminology, as it is applied to the whole
observation record to estimate the waveform at an intermediate time
\cite{smooth,*smooth_pra1,*smooth_pra2}.  By standard Wiener filtering
theory \cite{vantrees}, the minimum mean-square risk in the SPLOT
limit is
\begin{align}
\Avg{\mathsf R} &\to
\intall\frac{|\tilde h(\omega)|^2}
{|\tilde h_X(\omega)|^2/S_Z(\omega)+1/S_\theta(\omega)}
\frac{d\omega}{2\pi},
\label{R_waveform}
\\
\tilde h_X(\omega) &\equiv \intall h_X(t)\exp(-i\omega t) dt,
\end{align}
where $S_Z$ is the power spectral density of $Z$.  Comparing
Eqs.~(\ref{Qmax_waveform}) and (\ref{R_waveform}), one sees that
$\avg{\mathsf R} \ge \mathsf Q_{\rm max}$ implies
\begin{align}
\frac{S_Z(\omega)}{|\tilde h_X(\omega)|^2} 
&\ge \frac{\hbar^2}{4S_q(\omega)},
\end{align}
which serves as a fundamental quantum limit on the noise floor. To
reach this limit for an optomechanical system, backaction evasion and
quantum-limited measurements are necessary \cite{twc}. It is possible
to derive alternative quantum limits in terms of the optics by
appealing to the interaction picture and tighter limits that account
for loss by choosing the purification of the quantum state judiciously
\cite{tsang_open}.  Reference~\cite{iwasawa} reports an experimental
demonstration of mirror-motion estimation close to such quantum
limits.

It is noteworthy that, prior to Ref.~\cite{twc}, Braginsky and
coworkers derived an expression similar to Eq.~(\ref{K_waveform}) by
optimizing a signal-to-noise ratio (SNR) in terms of an observable
\cite{braginsky}.  A spectral form of their optimal SNR, derived from
a heuristic energy-time uncertainty relation, can be found in
Ref.~\cite{bgkt}. They called their results the energetic quantum
limit.  The similarity is not a coincidence, as the Helstrom
information can also be expressed as the solution to the optimization
problem 
\begin{align}
\mathsf K &= \max_Y \frac{(v^a\partial_a \bar Y)^2}
{\trace(Y - \bar Y)^2 \dop},
\label{SNR}
\\
\bar Y &\equiv \trace Y \dop,
\end{align}
where $Y$ is any observable and the right-hand side of Eq.~(\ref{SNR})
is similar to the SNR studied in
Ref.~\cite{braginsky}. Equation~(\ref{SNR}) can be proved by applying
the Cauchy-Schwarz inequality to
$(v^a\partial_a \bar Y)^2 = (\trace Y v^a\partial_a\dop)^2 = [\trace
(Y-\bar Y)v^a\partial_a\dop]^2 = \{\trace [(Y - \bar Y) \circ
(v^a\mathcal S_a)] \dop\}^2 \le [\trace (Y-\bar Y)^2\dop] [\trace
(v^a\mathcal S_a)^2 \dop]$.  While their results are seminal and
capture the basic physics, the results here and in Ref.~\cite{twc} are
more precise in terms of meaning. The SNR does not have a direct
operational meaning in statistics, whereas here the statistical
problem is clearly defined in terms of a mean-square risk, and the
bound is proven to hold for any POVM and any biased or unbiased
estimator, not just observables. The clear definition of a risk is
important, as different problems have different types of risk and
different optimal measurements, and no single SNR-based treatment can
deal with all of them.  For example, while a linear measurement in the
form of Eq.~(\ref{X_process}) can achieve the optimal SNR and also
optimal waveform estimation, more careful studies reveal that it is
suboptimal with respect to the quantum limits for waveform detection
\cite{tsang_nair} and spectrum parameter estimation \cite{ng16}, and
photon-counting measurements can perform much better for those
problems.

Equation~(\ref{Qmax_waveform}) demonstrates the importance of prior
information in the form of $1/S_\theta(\omega)$, as the integral may
not converge without it; see Ref.~\cite{berry13} for an example in
optical phase estimation. If $\beta = \theta(\tau)$,
Eqs.~(\ref{Qmax_waveform}) and (\ref{R_waveform}) are steady-state
values that do not scale with $T$. This is an extreme example where
the i.i.d.\ condition does not hold, the standard asymptotic theory
\cite{vaart,hayashi} fails, the convergence rate is slower than the
parametric rate, and prior information is indispensable.  The
information that can be acquired in one time slot with duration
$\delta t$ is infinitesimal, but a finite risk can still be achieved
because there exist prior correlations in $\theta(t)$ across different
times before and after $t = \tau$, meaning that information over
multiple time slots can contribute to the estimation of each
$\theta(\tau)$. This intuition explains why the optimal estimator is a
smoother.

\subsection{\label{sec_imaging}Subdiffraction incoherent optical
  imaging}

For another application of quantum estimation theory, consider the
far-field paraxial imaging of $p$ spatially incoherent and equally
bright point sources \cite{goodman}, as depicted in
Fig.~\ref{spade}. On the image plane, the density operator of each
photon can be modeled as \cite{tnl}
\begin{align}
\dop(\theta) &= \frac{1}{p}\sum_{a=1}^p
\exp\bk{-ik\theta^a}\ket{\Psi}\bra{\Psi} \exp\bk{ik \theta^a},
\\
\ket{\Psi} &= \intall dx \Psi(x)\ket{x},
\end{align}
where $\theta$ is a vector of the unknown source positions on the
object plane that is assumed to be one-dimensional for simplicity,
$\ket{x}$ is the Dirac eigenket for the image-plane photon position
that obeys $\braket{x|x'} = \delta(x-x')$, with an image-plane
coordinate $x$ that is normalized with respect to the magnification
factor, $\Psi$ is the point-spread function of the imaging system for
the optical field, and $k$ is the momentum operator.

\begin{figure}[htbp!]
\centerline{\includegraphics[width=0.45\textwidth]{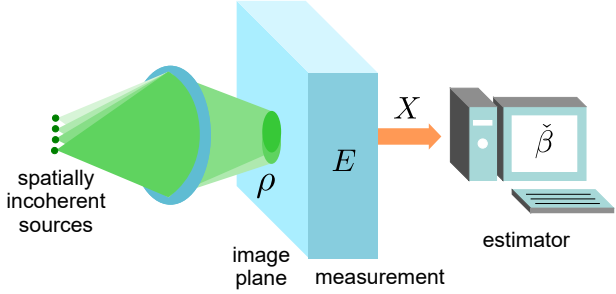}}
\caption{\label{spade}Basic setup of the optical imaging problem.}
\end{figure}

Direct imaging can be modeled as a measurement of each photon in the
position basis \cite{tnl}. The probability density of each observed
position is then
\begin{align}
f(x|\theta) &= \bra{x}\dop(\theta)\ket{x} = 
\frac{1}{p}\sum_{a=1}^p h(x-\theta^a),
\\
h(x) &\equiv \abs{\Psi(x)}^2.
\end{align}
The Fisher information is
\begin{align}
v^a F_{ab}(\theta) v^b &= 
\intall \frac{[v^a \partial_ah(x-\theta^a)]^2}{p^2f(x|\theta)} dx.
\end{align}
In particular, at $\theta = 0$,
\begin{align}
v^a F_{ab}(0) v^b &= \bk{v^a w_a}^2
\intall \frac{1}{h(x)} \Bk{\parti{h(x)}{x}}^2 dx,
\\
w_a &= \frac{1}{p},\quad a = 1,\dots,p.
\end{align}
The kernel of $F(0)$ is then the $(p-1)$-dimensional space
\begin{align}
\ker\Bk{F(0)} &= \BK{v \in \mathbb R^p: v^aw_a = 0},
\end{align}
while the range is the one-dimensional space
\begin{align}
\range\Bk{F(0)} &= \BK{c w: c \in \mathbb R}.
\end{align}
Assume hereafter that $\beta$ is a linear function of $\theta$, such
that $u$ is $\theta$-independent.  For any $\beta$ with
$u \notin \range[F(0)]$, a $v \in \ker\Bk{F(0)}$ can always be found
such that $v^au_a \neq 0$ but $v^a F_{ab}(0) v^b = 0$.
Section~\ref{sec_wave} then implies that, from the minimax
perspective, any estimator of this $\beta$ must have a convergence
rate slower than the parametric rate with respect to $n$ detected
photons. Only a $\beta$ with $u \in \range[F(0)]$ has a nonzero
information at $\theta =0$ for any $v$ with $v^au_a \neq 0$. This
$\beta$ is proportional to the object centroid
$w_a\theta^a = (\sum_a\theta^a)/p$, and the parametric rate is indeed
possible by taking the sample mean of the photon positions, provided
that $h$ has a finite variance \cite{tsang19b}.

For $p = 2$, other than the centroid, the second parameter may be
taken as the separation $|\theta^2-\theta^1|$ between the two
sources. Reference~\cite{tsang18} uses a special case of
Theorem~\ref{thm_rate} to prove that, since the exponent of the Fisher
information is $m = 2$ for $u \propto v \propto (1,-1)$, a limit on
the convergence rate is $\mathsf B_{\rm worst} = \mathit\Omega(n^{-1/2})$. This
rate is also observed numerically in
Refs.~\cite{tham17,tsang18}. Pa\'ur and coworkers showed that the
exponent can be improved to $m = 1$ if the point-spread function has
zeros \cite{paur18,*paur19}, and the limit becomes
$\mathsf B_{\rm worst} = \mathit\Omega(n^{-2/3})$ according to
Theorem~\ref{thm_rate}.

The Helstrom information turns out to be much higher
\cite{tnl,bisketzi19}. For $n$ detected photons and i.i.d.\ quantum
states, the Helstrom information is simply $n$ times that for one
photon \cite{hayashi}.  For $p = 2$, $K(\theta)$ turns out to be
full-rank \cite{tnl}, and separation estimation at the parametric rate
is also possible via spatial-mode demultiplexing \cite{tsang18}.  For
$p \ge 2$, Bisketzi and coworkers found that $K(\theta)$ has a rank of
two as $\theta \to 0$ \cite{bisketzi19}. Then Sec.~\ref{sec_wave}
implies that any $\beta$ with a $u \notin \range[K(0)]$ cannot be
estimated at the parametric rate by any measurement, and only a
$\beta$ with $u$ in the two-dimensional range may be estimated at the
parametric rate.





\section{Conclusion}
Compared with the local theory, the use of Bayesian Cram\'er-Rao
bounds has been less systematic in the literature and often relied on
the ingenuity of the researcher to pick the appropriate form. This
work resolves some of the ambiguities and hopefully inspires further
progress via the physics connections. 

The formalism here looks ripe for a generalization for
infinite-dimensional parameter spaces in a manner similar to the local
theory \cite{bickel93,tsang20,stein56}.  An important application
would be to derive semiparametric bounds with slow convergence rates
\cite{gill95} in a more systematic fashion.

\section*{Acknowledgment}
This work is supported by the National Research Foundation (NRF)
Singapore, under its Quantum Engineering Programme (Grant No.~QEP-P7).

\appendix
\section{\label{sec_vectoral}Vectoral parameter of interest}
Here I generalize the fundamental results in
Secs.~\ref{sec_invar}--\ref{sec_quantum} for a vector parameter of
interest $\beta(\theta) = (\beta^1,\dots,\beta^q) \in \mathbb R^q$
with $1 \le q \le p$. Define the mean-square risk as
\begin{align}
\mathsf R(\theta) &\equiv 
\int \Bk{\check\beta^j(x) - \beta^j(\theta)}
\gamma_{jk}(\theta)\Bk{\check\beta^k(x) - \beta^k(\theta)}
\nonumber\\&\quad
\times f^{(n)}(x|\theta) d\mu^{(n)}(x),
\end{align}
where $\gamma$ is a positive-definite weight matrix.  For clarity,
indices starting from $j$ are used to label the components of $\beta$,
to be distinguished from indices that start from $a$ for the
components of $\theta$.  The Bayesian risk is
\begin{align}
\Avg{\mathsf R} &= \int \mathsf R(\theta)\pi(\theta) d^p\theta.
\end{align}
Define
\begin{align}
u_a^j &\equiv \partial_a \beta^j.
\end{align}
The Gill-Levit bounds $\mathsf B$ still have the form of
Eq.~(\ref{GL}), but now \cite{gill95}
\begin{align}
\mathsf A &\equiv v_j^au_a^j,
\\
\mathsf F &\equiv \gamma^{jk} v_j^a F_{ab} v_k^b,
\\
\mathsf P &\equiv \gamma^{jk}\Bk{\frac{1}{\pi}\partial_a\bk{\pi v_j^a}}
\Bk{\frac{1}{\pi}\partial_b\bk{\pi v_k^b}},
\end{align}
where $v$ now has $q\times p$ entries and
\begin{align}
\gamma^{jk} &\equiv (\gamma^{-1})^{jk}.
\end{align}
$\{\gamma^{jk}\}$ are the entries of the $B$ matrix in
Ref.~\cite{gill95}, while $\{v_j^a\}$ are the entries of the $C$
matrix in Ref.~\cite{gill95}.

Upon reparametrization of $\theta$, $\gamma$ should remain invariant,
in the sense of $\gamma(\theta) = \tilde\gamma(\tilde\theta(\theta))$,
so that the statistical problem remains unchanged.

It is straightforward to generalize Prop.~\ref{prop_contra}.
\begin{proposition}
  $\mathsf B$ is invariant under reparametrization if $v_j$ for each
  $j$ obeys the transformation law
\begin{align}
v_j^a J_a^b &= \tilde v_j^b.
\end{align}
\end{proposition}
\begin{proof}
  Almost identical to that of Prop.~\ref{prop_contra} and omitted for
  brevity.
\end{proof}
For a manifold $\Theta$, a generalization of
Prop.~\ref{prop_invariant} is as follows.
\begin{proposition}
  If $\prior v$ vanishes on any boundary of $\Theta$, the Bayesian
  mean-square risk has a lower bound given by Eq.~(\ref{GL}), where
\begin{align}
\Avg{\mathsf A} &\equiv \int \bk{v_j^a u_a^j}  \prior \epsilon,
\label{Avec}
\\
\Avg{\mathsf F} &\equiv \int \bk{\gamma^{jk}v_j^a F_{ab} v_k^b} \prior \epsilon,
\\
\Avg{\mathsf P} &\equiv 
\int \gamma^{jk}\Bk{\frac{1}{\prior}\nabla_a\bk{\prior v_j^a}}
\Bk{\frac{1}{\prior}\nabla_b\bk{\prior v_k^b}} \prior\epsilon.
\label{Pvec}
\end{align}
\end{proposition}
\begin{proof}
Let 
\begin{align}
\mathsf b^j &\equiv \int \bk{\check\beta^j - \beta^j} f^{(n)} d\mu^{(n)}.
\end{align}
By the Leibniz rule,
\begin{align}
\int \nabla_a\bk{\mathsf b^j \prior v_j^a} \epsilon
&= \iint \bk{\check\beta^j-\beta^j} \nabla_a\bk{f^{(n)}\prior v_j^a}
d\mu^{(n)}\epsilon
\nonumber\\&\quad
- \int \bk{v_j^a\nabla_a\beta^j} \prior \epsilon.
\end{align}
The left-hand side is zero by the Stokes theorem, 
if $\prior v$ vanishes on any boundary of $\Theta$. Then
\begin{align}
\Avg{\mathsf A} &= 
\expect\Bk{\bk{\check\beta^j-\beta^j} \gamma_{jk} s^k },
\label{A_expect}
\\
s^j &\equiv \frac{\gamma^{jk}}{f^{(n)}\prior}\nabla_a\bk{f^{(n)}\prior v_k^a}.
\end{align}
Considering the right-hand side of Eq.~(\ref{A_expect}) as an inner
product between $(\check\beta-\beta)$ and $s$ that is weighted by
$\gamma$ and applying the Cauchy-Schwarz inequality, I obtain
\begin{align}
\Avg{\mathsf A}^2 &\le \Avg{\mathsf R}\expect\bk{s^j \gamma_{jk} s^k}.
\end{align}
Standard procedures then lead to Eqs.~(\ref{GL}) and
(\ref{Avec})--(\ref{Pvec}).
\end{proof}

A generalization of Theorem~\ref{thm_optimal} is as follows.
\begin{theorem}
\label{thm_optimal_vec}
\begin{align}
\max_v \mathsf B = \Avg{u,L^{-1}u}_\prior \equiv \mathsf B_{\rm max},
\end{align}
where the inner product is defined as
\begin{align}
\Avg{v,u}_\prior &\equiv \int \bk{v_j^{a}u_a^j} \prior \epsilon,
\end{align}
the linear, self-adjoint, and positive-semidefinite operator $L$ is
defined as
\begin{align}
(Lv)_a^j &\equiv
n \gamma^{jk} F_{ab} v_k^{b} -\nabla_a\Bk{\frac{\gamma^{jk}}{\prior}\nabla_b\bk{\prior v_k^b}},
\label{L_vec}
\end{align}
and $u$ is assumed to be in the range of $L$, such that $L^{-1}u$
exists. A least favorable $v$ that maximizes $\mathsf B$ must satisfy
\begin{align}
v &\propto L^{-1} u.
\end{align}
\end{theorem}
\begin{proof}
Similar to that of Theorem~\ref{thm_optimal} and omitted for brevity.  
\end{proof}
With the substitution $\prior= \psi^2$, Eqs.~(\ref{A_wave})--(\ref{D})
can be generalized to
\begin{align}
\Avg{\mathsf A} &= \int \bk{v_j^a u_a^j} \psi^2\epsilon,
\\
\Avg{\mathsf F} &= \int  \bk{\gamma^{jk} v_j^a F_{ab} v_k^b}\psi^2\epsilon,
\\
\Avg{\mathsf P} &= \int \gamma^{jk}\bk{\mathsf D_j\psi}\bk{\mathsf D_k\psi}\epsilon
\\
\mathsf D_j\psi &\equiv \bk{\nabla_a v_j^a}\psi + 2 v_j^a\nabla_a\psi,
\end{align}
while Eqs.~(\ref{P_wave2}), (\ref{D_dagger}), and (\ref{H})
can be generalized to
\begin{align}
\Avg{\mathsf P} &= \Avg{\psi,\mathsf D_j^\dagger\bk{\gamma^{jk} \mathsf D_k\psi}},
\\
\mathsf D_j^\dagger \psi &= \bk{\nabla_a v_j^a}\psi - 2 v_j^a\nabla_a\psi,
\\
\mathsf H \psi &\equiv  
n\mathsf F \psi + \mathsf D_j^\dagger\bk{\gamma^{jk} \mathsf D_k \psi}.
\end{align}
If a parametrization with $g_{ab} = \delta_{ab}$ can be assumed and
$v$ is $\theta$-independent, a further simplification is
\begin{align}
\mathsf H \psi &= n \mathsf F \psi -
4v_j^a\partial_a\bk{\gamma^{jk}v_k^b\partial_b\psi}.
\end{align}
The last term becomes the Laplacian $\partial_a\partial^a \psi$ if
$q=p$ and $v$ and $\gamma$ are assumed to be identity
matrices.

To apply the preceding results to quantum problems, Eq.~(\ref{nagaoka})
can be generalized to
\begin{align}
\mathsf F &= \gamma^{jk} v_j^a F_{ab} v_k^b \le 
\gamma^{jk} v_j^a K_{ab} v_k^b,
\end{align}
as both $K_{ab}-F_{ab}$ and $\gamma^{jk} v_j^a v_k^b$ are
positive-semidefinite.  Equations~(\ref{Qmax}) and (\ref{R}) can then
be generalized by redefining $R$ as
\begin{align}
(Rv)_a^j &\equiv
\gamma^{jk} K_{ab} v_k^{b} 
-\nabla_a\Bk{\frac{\gamma^{jk}}{\prior}\nabla_b\bk{\prior v_k^b}}.
\end{align}

\bibliography{research2}

\end{document}